\documentclass[journal]{IEEEtranTIE}
\usepackage{graphicx}
\usepackage{cite}
\usepackage{picinpar}
\usepackage{flushend}
\usepackage[latin1]{inputenc}
\usepackage{colortbl}
\usepackage{multirow} 
\usepackage{soul}
\usepackage{pifont}

\usepackage{alltt}
\usepackage[hidelinks]{hyperref} 
\usepackage{enumerate}
\usepackage{siunitx}
\usepackage{breakurl}
\usepackage{epstopdf}
\usepackage{amsmath,amsfonts}
\usepackage{algorithmic}
\usepackage{array}
\usepackage{amssymb}
\usepackage{textcomp}
\usepackage{booktabs}
\usepackage{stfloats}
\usepackage{url}
\usepackage{bm}
\usepackage{verbatim}
\usepackage{pbox}
\usepackage{makecell}
\usepackage{subcaption}
\usepackage{setspace}
\usepackage{float}

\newtheorem{theorem}{Theorem}
\newtheorem{assumption}{Assumption}
\newtheorem{remark}{Remark}

\newtheorem{proof}{Proof}
\newtheorem{problem}{Problem}

\newcommand{\Int}{{\rm {Int}}}

\newcommand{\reals}{\mathbb{R}}

\newcommand{\f}{\bar{f}}
\newcommand{\g}{\bar{g}}

\newcommand{\x}{\bar{\bm{x}}}
\newcommand{\bmd}{\bar{\bm{d}}}

\newcommand{\dtau}{{\rm d}\tau}
\newcommand{\K}{\mathcal{K}}
\newcommand{\sign}{{\rm{sign}}}
\newcommand{\classK}{class-$\K$ }

\newcommand{\D}{\partial}
\def\BibTeX{{\rm B\kern-.05em{\sc i\kern-.025em b}\kern-.08em
    T\kern-.1667em\lower.7ex\hbox{E}\kern-.125emX}}
\usepackage{balance}

\begin{document}
\setlength{\abovedisplayskip}{2pt}   
\setlength{\belowdisplayskip}{2pt}
\setlength{\abovedisplayshortskip}{2pt}
\setlength{\belowdisplayshortskip}{2pt}
\title{Robust Safety Critical Control Under Multiple State and Input Constraints: Volume Control Barrier Function Method}
\author{\vskip 1em
	Jinyang Dong,
	Shizhen Wu,
     Rui Liu,
	Xiao Liang, \emph{Senior Member,~IEEE},
	\\Biao Lu, \emph{Member,~IEEE},
    and Yongchun Fang, \emph{Senior Member,~IEEE}
    \thanks{
        This work is supported in part by the National Natural Science Foundation of China under Grant 62233011; the Natural Science Foundation of Guangdong Province under Grant 2025A1515011967;
        and in part by the Key
Technologies R \& D Program of Tianjin under Grant 23YFZCSN00060.   \emph{(Corresponding author: Yongchun Fang)}

		
		The authors are with the Institute of Robotics and Automatic Information System, College of Artificial Intelligence, Nankai University, Tianjin 300071, China, and Institute of Intelligence Technology and Robotic Systems, Shenzhen Research Institute of Nankai University, Shenzhen 518083, China.  (e-mail: \url{dongjinyang@mail.nankai.edu.cn}; \url{szwu@mail.nankai.edu.cn}; \url{1911635@mail.nankai.edu.cn}; \url{liangx@nankai.edu.cn}; \url{lubiao@nankai.edu.cn}; \url{fangyc@nankai.edu.cn}).
	}
}

\maketitle
\begin{spacing}{1}	
\begin{abstract}
In this paper, the safety-critical control problem for uncertain systems under multiple control barrier function (CBF) constraints and input constraints is investigated.
A novel framework is proposed to generate a safety filter that minimizes changes to reference inputs when safety risks arise, ensuring a balance between safety and performance.
A nonlinear disturbance observer (DOB) based on the robust integral of the sign of the error (RISE) is used to estimate system uncertainties, ensuring that the estimation error converges to zero exponentially.
This error bound is integrated into the safety-critical controller to reduce conservativeness while ensuring safety.
To further address the challenges arising from multiple CBF and input constraints, a novel Volume CBF (VCBF) is proposed by analyzing the feasible space of the quadratic programming (QP) problem. 
To ensure that the feasible space does not vanish under disturbances, a DOB-VCBF-based method is introduced, ensuring system safety while maintaining the feasibility of the resulting QP.
Subsequently, several groups of simulation and experimental results are provided to validate the effectiveness of the proposed controller.
\end{abstract}

\markboth{IEEE TRANSACTIONS ON INDUSTRIAL ELECTRONICS}%
{}

\definecolor{limegreen}{rgb}{0.2, 0.8, 0.2}
\definecolor{forestgreen}{rgb}{0.13, 0.55, 0.13}
\definecolor{greenhtml}{rgb}{0.0, 0.5, 0.0}

\section{Introduction}

\IEEEPARstart {A}{s} automation systems have become integral to our daily lives, the development of safe and high-performance controllers for these systems is of paramount importance.
To meet this need, the Control Barrier Function (CBF) is a powerful tool to ensure the safety of control systems\cite{ames2019control}.
The CBF-based controllers are typically implemented in the form of Quadratic Programming (QP), where CBF conditions and actuator limits are incorporated as constraints.
Due to the low computational complexity of CBFs, they are successfully applied in various automation and robotic systems\cite{xie2025certificated,10505850,9811271,10601510}.


The modular design of CBF allows multiple conditions to be incorporated into the QP problem to handle multi-state constraints.
However, the feasibility of the QP problem heavily relies on the compatibility of CBF conditions, and potential conflicts among constraints make the design process highly challenging.
As multiple constraints may lead to conflicts, ensuring the compatibility of CBF conditions is a critical challenge in the design process \cite{saveriano2019learning,10363340}.
To address this issue, some studies attempt to synthesize multiple CBFs into one.
For instance, \cite{7937882} and \cite{9247270} adopt different approaches, merging multiple barrier functions using Boolean operators, all of which may result in performance degradation (including undesirable oscillatory behavior)\cite{black2023consolidated}.
An adaptive method that synthesizes multiple CBFs into a candidate CBF with parameters updated online is proposed in \cite{10383597}.
However, the optimization problem underlying this adaptation law must satisfy Slater's condition, which limits its flexibility.
In \cite{parwana2023feasible}, H. Parwana et al. define the volume of the feasible space as the volume CBF (VCBF) to address feasibility issues arising from conflicts between state and input constraints.
Compared with other solutions, it is easy to implement, suitable for scenarios with high real-time requirements, and exhibits good adaptability in high-dimensional state spaces or complex input constraints, reducing sensitivity to control parameters.
It is well established that disturbances are unavoidable in practical implementations \cite{emam2019robust, 10230907,10591233}.
To the best of our knowledge, how to ensure the feasibility of QP control laws under multi-state and input constraints in the presence of disturbances remains an area requiring further investigation.
Therefore, this paper will investigate the problem of robust safety critical control for uncertain systems under multiple state and input constraints.
Disturbance observers (DOB), as an effective tool for estimating disturbances based on system dynamics and measurable states, are widely applied in robotics and automotive systems \cite{10403533,8970337,9447165}.
Building on this, this paper argues that by effectively leveraging the strengths of VCBF and DOB, it is possible to develop a promising solution to ensure the feasibility of QP control laws under multi-state and input constraints in the presence of disturbances.
However, this approach faces notable challenges.
First, the VCBF proposed in \cite{parwana2023feasible} is specifically designed for disturbance-free scenarios and may fail when disturbances are present.
Second, existing DOB-CBF methods assume that estimation errors eventually converge to a bounded set \cite{10156095,10388430,10617810}.
Combining these methods with the approach in \cite{parwana2023feasible} to construct the feasible space could introduce a certain degree of conservativeness.


To address these challenges, this paper aims to further develop the VCBF method proposed in \cite{parwana2023feasible}.
In this paper, the design of a robust controller for uncertain systems subject to multiple state and input constraints is investigated, and a novel framework is proposed to ensure safety while balancing system performance.
Inspired by the results of \cite{patil2021exponential}, an observer based on the robust integral of the sign of the error (RISE) is designed to estimate disturbances, guaranteeing that the estimation error exponentially converges to zero.
Utilizing this property, the time-varying lower bound of the estimation error is integrated into the QP-based controller, which ensures safety while greatly reducing the conservativeness of the safe input set.
To further address challenges arising from multiple CBF constraints as well as input constraints, 
the VCBF is incorporated into the QP formulation.
A novel control law is designed to prevent conflicts between multiple CBFs even in the presence of disturbances. 
Finally, to validate the proposed method, various simulation and experiments are conducted.
To the best of our knowledge, this paper is the first to experimentally verify the feasibility guarantee control method based on VCBF.
The main contributions of this paper are as follows:
\begin{enumerate}
\item
Inspired by the latest RISE-based DOB research \cite{patil2021exponential}, which for the first time proposed an observer with exponential convergence rates, we leverage time-varying bounds of the estimated disturbance and estimation error to replace the unknown terms in the high order CBF (HOCBF). This approach ensures safety while reducing the conservativeness of the safe input set.
\item
It further proposes a DOB-VCBF-based method to ensure that the feasible space does not vanish under disturbances. Furthermore, the proposed control law is further integrated with the DOB-HOCBF method and solved through a QP problem, ensuring system safety while maintaining compatibility between input constraints and multiple CBF constraints.
\item Extensive simulation and experiments are conducted to validate the effectiveness of the proposed controller, demonstrating its ability to ensure safety and maintain robust performance under uncertain conditions.
\end{enumerate}

The remaining parts of this paper are organized as follows.
After the preliminaries are introduced in Section \ref{sec:b}, Section \ref{sec:c} discusses the problem of robust safe control with multiple state and input constraints.
Section \ref{sec:d} analyzes the RISE-based observer and designs a DOB-VCBF-QP based control law.
Subsequently, several groups of simulation and experimental results are given to validate the effectiveness of the proposed controller in Section \ref{sec:e}.
Finally, Section \ref{sec:f} summarizes this paper.

\emph{Notations:} The real number is represented as $\reals$. $\reals^+_0$ denotes the nonnegative real number.  $[\rm N]$ denotes the set of integers $\{1, 2, \dots, N\}$.
$\Gamma(\cdot)$ is the $\Gamma$ function.
The interior and boundary of a set $\mathcal{C}$ are indicated by $\textrm{Int}(\mathcal{C})$ and $\partial \mathcal{C}$.
The empty set is denoted by $\emptyset$.
A continuous function $\alpha:\reals\rightarrow\reals$ belongs to extended \classK ($\alpha\in\K_e$) if it is strictly increasing and $\alpha(0)=0, \alpha(r)\rightarrow+\infty$ as $r\rightarrow+\infty$, $\alpha(r)\rightarrow-\infty$ as $r\rightarrow-\infty$.
$\|y\|$ represents the Euclidean norm of $y$.
$\sign(\cdot)$ denotes the symbolic function.
For the system $\dot{\bm x}=f(\bm x)+g(\bm x)\bm u$, the Lie derivatives of $h(\bm x)$ are $L_{f}h(\bm x)=\frac{\D h}{\D \bm x}f(\bm x)$ and $L_gh(\bm x)=\frac{\D h}{\D \bm x}g((\bm x)$. Specifically, when $f(\bm x)$ is the identity matrix $\bm I$, the Lie derivative of is denoted by $L_Ih(\bm x)$.
\end{spacing}
\begin{spacing}{1}
\section{Preliminaries}\label{sec:b}
Consider a nonlinear, control-affine system with unknown disturbance of the form
\begin{align}
  \label{equ:sys1}
  \dot{\bm x}&=f(\bm x)+g(\bm x)\bm u+\bm d,
\end{align}
where $\bm x \in\bm X\subset\mathbb{R}^n$ is the state, $\bm u \in \bm U\subset\mathbb{R}^m$ is the control input, $\bm d\in \bm D\subset\reals^n$ represents the unknown disturbance.
$f:\bm X \to \mathbb{R}^n$ and $g:\bm X \to \mathbb{R}^{n\times m}$ are sufficiently smooth functions.
The sets $\bm U$ and $\bm D$ are compact.
Specifically, $\bm U$ is defined as $\bm U=\{\bm u\in \reals^m|\bm A_m\bm u\leq \bm b_m\}$, where $\bm A_m\in\reals^{2m\times m}$, $\bm b_m\in\reals^{m}$ are both constant matrices.
Under this definition, $\bm U$ forms a convex polytope.
To simplify the description, the time $t$ is considered as an extended state. By defining the extended state vector $\bar{\bm x}=[\bm x^\top,t]^\top$, the system (\ref{equ:sys1}) can be rewritten as
\begin{align}
  \label{equ:sys2}
  \dot{\bar{\bm x}} &= \underbrace{\left[\begin{array}{c} f(\bm x) \\ \hline 1 \end{array}\right]}_{\f(\bar{\bm x})} + \underbrace{\left[\begin{array}{c} g(\bm x) \\ \hline \bm 0_{1\times m} \end{array}\right]}_{\g(\bar{\bm x})} \bm u+\bar{\bm d},
\end{align}
where $\x\in \bar{\bm X}\subset\bm X\times \reals^+$, $\bmd=[\bm d^\top,0]^\top$.

\subsection{Control Barrier Function}
Suppose that the safe set $\bm S$ is defined as the zero-superlevel set of a continuously differentiable function $h(\bar{\bm x})$:
\begin{subequations}
 \label{equ:Si}
    \begin{align}
        \bm S &\triangleq \{\x\in \bar{\bm X}: h(\bar{\bm x}) \geq 0 \}, \\
        \Int \bm S &\triangleq \{\x\in \bar{\bm X}: h(\bar{\bm x}) > 0 \}, \\
        \partial \bm S &\triangleq \{\x\in \bar{\bm X}: h(\bar{\bm x}) = 0 \}.
    \end{align}
\end{subequations}
Additionally, we assume that there exists a locally Lipschitz continuous controller $\bm u=\bm \pi(\x)$, where $\bm \pi:\bar{\bm X}\to\reals^m$.
When the disturbance does not exist, i.e., $\bmd=\bm 0$,
the controller $\bm\pi$ keeps (\ref{equ:sys2}) safe with respect to $\bm S$, if for any initial state $\x(0)\in\bm S$, the solution to (\ref{equ:sys2}), defined as $\x(t)$, remains within the safe set $\bm S$.

\emph{Definition 1 (CBF):}
Let $\bm S$ be the zero-superlevel set of a continuously differentiable function $h(\bar{\bm x})$ with $\frac{\D h}{\D x}\neq 0$ when $h(\bar{\bm x})=0$.
Then, $h(\bar{\bm x})$ is a CBF for system (\ref{equ:sys2}) on $\bm S$ if there exists an extended $\mathcal{K}$ function $\alpha$ satisfying
\begin{align}
  \label{equ:limit_U}
  \sup_{u\in \bm U}\!\Big[\!L_{\bar f}h(\x)\!\!+\!\!L_{\bar g}h(\x)\bm u\Big]\!\geq\!-\alpha(h(\x)).\!
\end{align}

However, there may be situations where $L_{\bar g}h(\x)=0$ in practical applications, that is, $\bm u$ does not show up in (\ref{equ:limit_U}). Therefore,  this CBF constraint cannot be used to formulate an optimization problem directly.
In \cite{tan2021high}, high-order CBF (HOCBF) is introduced as a way to enforce high relative degree safety constraints. The $r$th-order differentiable function $h(\x)$ is said to have input relative degree $r$ if for all $\x\in\bar{\bm X}$:
\begin{align}
    L_{\bar g}L_{\bar f}^{r-1}h(\x)&\!\neq\!0,\nonumber\\
    L_{\bar g}h(\x)= L_{\bar g}L_{\bar f}h(\x)=\!\cdots\!= L_{\bar g}L_{\bar f}^{r-2}h(\x)&\!=\! 0.\nonumber
\end{align}
Then, for a differentiable function $h(\x)$, consider a sequence of functions:
\begin{align}
    \label{equ:highorder_phi}
   \phi^0(\x)&=h(\x),\nonumber\\
   \phi^k(\x)&=\dot{\phi}^{k-1}(\x)+\alpha^k(\phi^{k-1}(\x)),~k\in[r-1],
\end{align}
where $\alpha^k\in\K$ is the $(r-k)$th differentiable function.
The associated extended sets $S^{k-1}, k\in[r]$ is further defined as:
\begin{align}
\label{equ:highorder_Sik}
    \bm S^{k-1}\triangleq\{\x\in \bar{\bm X}: \phi^{k-1}(\bar{\bm x}) \geq 0 \}.
\end{align}
Safety with respect to set $\bigcap_{k = 0}^{r-1}\bm S^k$ can be guaranteed through the HOCBF.

\emph{Definition 2 (HOCBF):}
Let $\phi^{k-1}(\x),k\in[r]$ be defined as (\ref{equ:highorder_phi}) and $ \bm S^{k-1} ,k\in[r]$ be defined as (\ref{equ:highorder_Sik}).
$h(\x)$ is called a HOCBF for system (\ref{equ:sys2}) with input relative degree $r$ on the set $\bigcap_{k = 0}^{r-1}\bm S^k$ if there exist $(r-k)$th order differentiable \classK functions $\alpha^k, k\in[r-1]$ and a \classK function $\alpha^{r}$ satisfying:
\begin{align}
  \label{equ:highorder_limit_U}
  \!\!\sup_{u\in \bm U}\!\Big[\!L_{\bar f}^{r}h(\x)\!+\!\!L_{\bar g}\!L_{\bar f}^{r\!-\!1}h(\x)\bm u+\!\mathcal{O}(h(\x)\!)\!\Big]\!\!\geq\!-\alpha^{r}(\phi^{r\!-\!1}(\x)),
\end{align}
where $\mathcal{O}(h(\x))=\Sigma_{k=1}^{r-1}L_{\bar f}^k(\alpha^{r-k}\circ\phi^{r-k-1})(\x)$.

In this paper, $\alpha^{k}, k\in[r]$ is chosen as a linear function for convenience. Thus, (\ref{equ:highorder_limit_U}) can be rewritten as
\begin{align}
  \label{equ:highorder_limit_U1}
  \sup_{u\in \bm U}\!\Big[\Gamma(\x,\bm u)\!=\!L_{\bar f}^{r}h(\x)\!+\!L_{\bar g}\!L_{\bar f}^{r\!-\!1}\!h(\x)\bm u\Big]\geq-\bm K^\top\bm\eta(\x),
\end{align}
where $\bm\eta(\x)=[L_{\bar f}^{r-1}h(\x), L_{\bar f}^{r-2}h(\x),\dots, h(\x)]^\top$, $\bm K=[k^1,\dots,k^{r}]^\top$ is the parameter chosen such that the roots  of $\lambda^{r}+k^1\lambda^{r-1}+\cdots+k^{r}=0$ are all negative reals $-\lambda_i<0,i\in[r]$.
Considering an HOCBF $h(\x)$, the set of control inputs satisfying the HOCBF condition is defined as:
\begin{align}
  \label{equ:khocbf}
  {\bm U}_{{\rm H}}(\x)=\{&\bm u\in \bm U: \Gamma(\x,\bm u)\geq-\bm K^\top\bm\eta(\x)\}.
\end{align}

\emph{Lamma 1 \cite{tan2021high}:}
Given an HOCBF $h(\x)$ from Definition 2 with the associated sets $\bm S^{k-1} ,k\in[r]$ defined by (\ref{equ:highorder_Sik}), any Lipschitz continuous controller $\bm u\in{\bm U}_{{\rm H}}(\x)$ will render the set $\bm S^{k-1} ,k\in[r]$ forward invariant for the system (\ref{equ:sys2}).


\subsection{RISE-based Disturbance Observer}
In practical systems, there exist unknown disturbance $\bmd$, which presents inequalities (\ref{equ:limit_U}) and (\ref{equ:highorder_limit_U1}) from being directly applicable.
A feasible approach to address this problem is to estimate the uncertainty and use a quantified error bound to manage the estimation error.
\begin{assumption}\label{assumption:bounded_continuity}
The disturbance $\bmd$ is assumed to be continuous, and its first and second derivatives $\dot{\bmd}$ and $\ddot{\bmd}$ are also continuous and bounded. Specifically, there exist positive constants $\delta_1$, $\delta_2$ and $\delta_3$ such that:
\[
\|\bmd\| \leq \delta_1, \quad \|\dot\bmd\| \leq \delta_2, \quad \|\ddot{\bmd}\| \leq \delta_3.
\]
\end{assumption}

Motivated by \cite{patil2021exponential} and \cite{10413553}, the RISE-based observer can be used to estimate the unknown disturbances.
Let $\hat{\x}$ and $\hat{\bmd}$ be the estimation of $\x$ and $\bmd$. The estimation errors are defined as $\tilde \x=\x-\hat\x$ and $\tilde{\bmd}=\bmd-\hat{\bmd}$.
The RISE-based disturbance observer can be designed as
\begin{align}
  \label{equ:obs}
  \dot{\hat{\x}} &=  \bar f(\x)+\bar g(\x)\bm u+\hat{\bmd}+\alpha\tilde\x, \\
  \label{equ:obs1}
  \dot{\hat{\bmd}} &=  \gamma(\dot{\tilde{\x}}+\alpha\tilde{\x})+\tilde\x+\beta{\rm sign}(\tilde\x),
\end{align}
where $\alpha, \beta$ and $\gamma$ are positive control gains. The initial estimation states are set as $\hat\x(0)=\x(0)$ and $\bmd(0)=0$.

\emph{Lamma 2:} Consider the uncertain system (\ref{equ:sys2}) with a continuously differentiable function $\bmd$ that satisfies Assumption \ref{assumption:bounded_continuity}. 
If the control gains $\alpha, \beta$ and $\gamma$ in  (\ref{equ:obs}) and (\ref{equ:obs1}) are chosen as $\beta>\delta_2+\delta_3/\max\{1,\alpha-\gamma\}$, then the estimation error $\tilde{\bmd}$ is exponentially convergent in the sense that $\|\tilde{\bmd}(t)\|\leq\delta_1 e^{-\lambda_V t/2}$, $\forall t\geq0$, with $\lambda_V=\min\{\alpha-1,\gamma\}$.
\begin{proof}
Given the definition of estimation errors, taking the time derivative of $\tilde \x$ and $\tilde{\bmd}$ and substituting (\ref{equ:sys2}), (\ref{equ:obs}) and (\ref{equ:obs1}) yields
\begin{align}
  \label{equ:tildex}
  \dot{\tilde{\x}} &= \tilde{\bmd}-\alpha\tilde\x, \\
  \label{equ:tilded}
  \dot{\tilde{\bmd}} &= \dot{\bmd}-\gamma\tilde{\bmd}-\tilde\x-\beta{\rm sign}(\tilde\x).
\end{align}
To facilitate the construction of a candidate Lyapunov function for analyzing the stability and convergence properties of $\tilde{\bmd}$, the following auxiliary variables are introduced
\begin{align}
  \label{equ:auxv1}
  \dot{P}&= -\kappa P-L, \\
  \label{equ:auxv2}
  L &= \tilde{\bmd}(-\beta\sign(\tilde\x)+\dot{\bmd}),
\end{align}
where $\kappa\in\reals^+$ is an auxiliary constant. The initial value of $P$ is set as $P(0)=0$.
Solving the ordinary differential equation (\ref{equ:auxv1}) yields
\begin{align}
  \label{equ:P1}
  P(t) &= -e^{-\kappa t}\int_0^t\tilde{\bmd}(-\beta\sign(\tilde\x)+\dot{\bmd}) e^{-\kappa \tau}\dtau.
\end{align}
Then, substituting (\ref{equ:tildex}) into (\ref{equ:P1}), one can obtain
\begin{align}
  \!\!P(t) \!=\!& -\!\!e^{-\kappa t}\!\int_0^t\!(\dot{\tilde{\x}}(\tau)\!+\!\alpha\tilde{\x}(\tau)\!)(-\!\beta\sign(\tilde\x(\tau))\!+\!\dot{\bmd}(\tau)) e^{\kappa \tau}\dtau\nonumber\\
    =&(\tilde\x(0)\dot{\bmd}(0)\!-\!\beta\|\tilde\x(0)\|_1\!)e^{\!-\!\kappa t}\!+\!\beta\|\tilde\x(t)\|_1\!-\!\tilde\x(t)\dot{\bmd}(t)\nonumber\\
    &\!+\!\!\!\int_0^t\!\!\!(\!(\alpha\!-\!\!\kappa)(\beta\|\tilde\x(\tau)\|\!_1\!\!-\!\tilde\x(\tau)\dot{\bmd}(\tau)\!)\!\!+\!\!\tilde\x(\tau)\ddot{\bmd}(\tau)\!)
    e^{\kappa (\!\tau\!-\!t)}\!\dtau. \nonumber
\end{align}
Based on Assumption \ref{assumption:bounded_continuity}, it can be found that
\begin{align}
  \label{equ:P3}
  P(t) \geq\int_0^t((\alpha-\kappa)(\beta-\delta_2)-\delta_3))\|\tilde\x(\tau)\|_1e^{\kappa (\tau-t)}\dtau.
\end{align}
Therefore, if the control gains satisfy $\alpha>\kappa$ and $\beta>\delta_2+\delta_3/(\alpha-\kappa)$, $P(t)$ is nonnegative.
To simplify parameter selection, we choose $\kappa=\min\{\alpha-1,\gamma\}$.
Consider the Lyapunov function as follows:
\begin{align}
    \label{equ:V}
    V=\frac{1}{2}\tilde\x^\top\tilde\x+\frac{1}{2}\tilde{\bmd}^\top\tilde{\bmd}+P.
\end{align}
Taking the time derivative of (\ref{equ:V}) and substituting (\ref{equ:tildex}), (\ref{equ:tilded}) and (\ref{equ:auxv1}) yields
\begin{align}
    \label{equ:Vdot}
    \dot V=&\tilde\x^\top\dot{\tilde\x}+\tilde{\bmd}^\top\dot{\tilde{\bmd}}-\lambda P-L\nonumber\\
    =&-\alpha\|\tilde\x\|^2-\gamma\|\tilde{\bmd}\|^2-\kappa P\leq-\lambda_V V.
\end{align}
where $\lambda_V=\min\{\alpha,\gamma,\kappa\}=\min\{\alpha-1,\gamma\}$. From (\ref{equ:Vdot}), it follows that
\begin{align}
    \label{equ:Vdot1}
 &V(t)\leq V(0)e^{-\lambda_Vt},\nonumber\\
 \Rightarrow &\tilde\bmd\leq\sqrt{\tilde\x(0)^2+\tilde\bmd(0)^2+P(0)}e^{-\lambda_Vt/2}
\end{align}
Given that the initial values of $\tilde\x$ and $P$ are both set to zero, the above inequality simplifies to
\begin{align}
    \label{equ:dtilde}
    \|\tilde{\bmd}(t)\|\leq\|\tilde{\bmd}(0)\|e^{-\lambda_V t/2}.
\end{align}
According to assumption \ref{assumption:bounded_continuity}, it is assumed that $\|\bmd\| \leq \delta_1$.
If the initial value of $\hat\bmd(0)$ is set as zero, it follows that $\|\tilde\bmd(0)\|\leq\delta_1$.
Consequently, the inequality (\ref{equ:dtilde}) can be rewritten as
\begin{align}
    \label{equ:dtilde1}
    \|\tilde{\bmd}(t)\|\leq\delta_1e^{-\lambda_V t/2}.
\end{align}
This result shows that the estimation error $\|\tilde{\bmd}(t)\|$ decays exponentially over time.
\end{proof}
\begin{remark}
In some studies on robust output feedback control, observers are used to estimate system states and disturbances, with the estimation error proven to converge to zero \cite{8970337,9447165,9647024,9788485,10666861}.
Inspired by these works, this paper designs a RISE-based observer for robust safety critical control to observe and compensate for disturbances,
where $\hat\x$ can be regarded as an auxiliary variable.
\end{remark}
%

\section{Problem Formulation}\label{sec:c}\
In \cite{parwana2023feasible}, the paper discusses the feasibility issue of the QP problem with multiple CBFs for system (\ref{equ:sys2}) in the absence of disturbances.
Under such conditions, for system (\ref{equ:sys2}) with $N$ $r_i$th differentiable HOCBFs $h_i(\x)$, along with their associated sets $\bm S_i^{k-1}, i\in[r_i]$, the safety set $\bm S_{\rm H}$ is defined as $\bm S_{H}=\bigcap_{i = 1}^{N}\bigcap_{k = 0}^{r_i\!-\!1}\bm S_i^k$.
According to Lemma 1, any Lipschitz continuous controller satisfying the condition of
\begin{align}
  \bm u\in{\bm U}_{{\rm H}i}(\x)=\{&\bm u: \Gamma_i(\x,\bm u)\geq-\bm K_i^\top\bm\eta_i(\x)\},\nonumber
\end{align}
will ensure the forward invariance of the set $\bigcap_{k = 0}^{r_i\!-\!1}\bm S_i^k$.
Consequently, the forward invariance of the safety set $\bm S_{H}$ can be achieved using the following controller:
\begin{subequations}
 \label{equ:cbfqp}
\begin{align}
\label{equ:goal}
  &\bm u^*(\x) = \arg\min_{\bm u\in\bm U}\|\bm u-\bm u_{\rm ref}\|^2,\\
  \label{equ:constraints1}
  &{\rm{s.t.}}\quad\Gamma_i(\x,\bm u)\geq-\bm K_i^\top\bm\eta_i(\x), i\in[N],
\end{align}
\end{subequations}
where $\bm u_{\rm ref}$ is the reference input.
Given the constraints of the above optimization problem, define the feasible space as
\begin{align}
\bm U_{\rm F}(\x) = \bigcap_{i=1}^N {\bm U}_{{\rm H}i}(\x)\bigcap \bm U.
\end{align}
For the control-affine system (\ref{equ:sys2}), each $\bar {\bm U}_{{\rm H}i}(\x)$, $i\in[N]$ is affine in $\bm u$ and represents a half-space.
Together with the input constraints $\bm U$, the intersection of these constraints is a convex polytope.
For brevity, the feasible space is denoted as the following form
\begin{align*}
\bm U_{\rm F}(\x) = \{\bm u\in \reals^m|\bm A(\x)\bm u\leq\bm b(\x)\},
\end{align*}
where
\[
{\bm{A}}(\x)\!=\!\!
\begin{bmatrix}
\bm A_m \\
\!-L_{\bar g}\!L_{\bar f}^{r_1\!-\!1}\!h_1(\x)\! \\
\cdots \\
\!-L_{\bar g}\!L_{\bar f}^{r_N\!-\!1}\!h_N(\x)\!
\end{bmatrix}\!,
{\bm{b}}(\x)\!=\!\!
\begin{bmatrix}
\bm b_m \\
\!\bm K_1^\top\bm\eta_1(\x)\!+\!L_{\bar f}^{r_1}h(\x)\! \\
\cdots \\
\!\bm K_N^\top\bm\eta_N(\x)\!+\!L_{\bar f}^{r_N}h(\x)\!
\end{bmatrix}\!.
\]
Besides, each row of $\bm A(\x)$ is denoted by $\bm a_i$ and each element of $\bm b(\x)$ is denoted by $b_i$ where $i\in[2m+N]$.
If $\bm U_{\rm F}(\x)$ is empty, quadratic program (\ref{equ:cbfqp}) becomes infeasible.
To address this, the VCBF is introduced, where the volume of $\bm U_{\rm F}(\x)$ is defined as a new CBF, $V_v(\x)$.
By ensuring that $V_v(\x)$ always remains positive, the non-emptiness of $\bm U_{\rm F}(\x)$ is guaranteed.
It is important to note that directly computing the volume of a polytope is of much difficulty.
Instead, it can be approximated using the inscribed ellipsoid of the polytope. The largest inscribed ellipsoid of a polytope can be obtained by solving the following convex optimization problem:
\begin{align}
  \label{equ:eillipsoidqp1}
  &\min\quad\log\det \bm B^{-1}, \\
  &{\rm{s.t.}}\quad\|\bm B \bm a_i^\top\|+\bm a_i\bm d\leq b_i, i\in[2m+N].\nonumber
\end{align}
where $\bm B\in\reals^{m\times m}$ defines the shape and size of the ellipsoid, and $\bm d\in\reals^m$ represents its center.
Thus, the volume barrier function $V_v(\x)$ can be calculated as follows:
\begin{align}
  \label{equ:Vv}
  V_v(\x)=\frac{(\pi^{m/2})}{\Gamma(\frac{m}{2}+1)}\det \bm B-\xi,
\end{align}
where $\xi$ is a small positive quantity.
Define the compatible state space as $S_{V_v}$ satisfying $\{\x\in\bar{\bm X}|V_v(\x)\geq0\}$.
Any Lipschitz continuous controller $\bm u$ satisfied the condition of
\begin{align}
\label{equ:limitD}
{u\!\in\! \bm U_{\rm C}}(\x)\!=\!\{\bm u:\!L_{\bar f}V_v(\x)\!+\!L_{\bar g}V_v(\x)\bm u\!\!\geq\!\!-\lambda_{V_v}V_v(\x)\},
\end{align}
will ensure the forward invariance of the set $S_{V_v}$.
Consequently, to enforce the compatibility of CBF and input constraints and ensure that the system always remains within safe operational limits, the controller can be obtained by solving the following VCBF-QP:
\begin{subequations}
 \label{equ:qp}
\begin{align}
\label{equ:qp1}
  [\bm u^*\!(&\x), \delta^*\!(\x)]^\top \!=\! \arg\min\|\bm u\!-\!\bm u_{\rm ref}\|^2\!+\!M\delta^2, \\
  \label{equ:constraints1}
  {\rm{s.t.}}~&\bm A(\x)\bm u\leq \bm b(\x),\\
  \label{equ:constraints2}
  &L\!_{\bar f}\!V\!_v(\x)\!+\!L_{\bar g}\!V_v(\x)\bm u\!\geq\!\!-\lambda_{V_v}V_v(\x)\!+\!\delta,
\end{align}
\end{subequations}
where $\delta$ is a slack variable.

The aforementioned results do not take disturbances into consideration.
When disturbances are present, the form of function $\Gamma_i(\x,\bm u)$ is
\begin{align}
  \label{equ:Gamma}
  \Gamma_i(\x,\bm u)\!=\!L_{\bar f}^{r_i}h_i(\x)\!+\!L_{\bar g}\!L_{\bar f}^{r_i\!-\!1}\!h(\x)\bm u+\!L_{I}\!L_{\bar f}^{r_i\!-\!1}\!h(\x)\bmd,
\end{align}
which includes the unknown term $L_{I}\!L_{\bar f}^{r\!-\!1}\!h(\x)\bmd$.
To handle this challenge, it is recommended to estimate the disturbance, then quantify its corresponding estimation error bounds, and incorporate the impact of these results into the CBF constraints.
However, in scenarios involving multiple CBF constraints, a conservative estimation error bound can reduce the feasible space, potentially leading to infeasibility.
Therefore, ensuring the feasibility of multiple CBF constraints under disturbances requires addressing the following key issues:
\begin{problem}
\label{problem1}
 For HOCBF $h(\x)$, find an adaptive lower bounder $\Lambda(\x)$ for the unknown term $L_{I}L_{\bar f}^{r-1}h(\x)\bmd$, such that the difference between $\Lambda(\x)$ and $L_{I}L_{\bar f}^{r-1}h(\x)\bmd$ converges to zero as $t$ tends to infinity.
\end{problem}

The original VCBF, which is developed for disturbance-free scenarios, may become ineffective in the presence of disturbances.
In this case, it is important to propose a method that ensures the feasible space $\bm U_{\rm F}(\x)$ consistently exists, such that both the safety requirements defined by CBFs and the input constraints can be satisfied simultaneously.
\begin{problem}\label{problem2}
Consider the control-affine system (\ref{equ:sys2}) with unknown disturbance subject to $N$ CBF constraints $h_i(\x)$.
Given that the initial conditions satisfy $\x(0)\in\bm S_{H}$, and $\bm U_{\rm F}(\x(0))\neq\emptyset$, design a controller to ensure that $\x\in\bm S_{H}$, $\bm U_{\rm F}(\x)\neq\emptyset$ hold for all time.
\end{problem}
\section{Main Results}\label{sec:d}
In this section, a RISE-based disturbance observer is developed to estimate disturbances.
The estimated disturbance, along with a time-varying lower bound on the estimation error, is utilized to replace the unknown disturbance, providing a solution to Problem \ref{problem1}.
Furthermore, to address Problem \ref{problem2}, a novel DOB-VCBF-based approach is proposed to ensure the existence of feasible solutions for system (\ref{equ:sys2}) under multiple CBF and input constraints, even in the presence of disturbances.
\subsection{Robust Safety With Disturbance Estimation}\label{sec:RISE}
In this subsection, the DOB-CBF-QP-based safe controller is proposed for the HOCBF.
Based on the results of Lamma 2, the time-varying lower bound can be designed as
\begin{align}
\label{equ:lambdai}
\Lambda(\x)\!\!=\!\!L_{I}L_{\bar f}^{r-1}h(\x)\hat\bmd\!-\|L_{I}L_{\bar f}^{r-1}h(\x)\|\delta_1 e^{-\lambda_V t/2}.
\end{align}
In this expression, the first term compensates for unknown disturbances, while the second term mitigates the impact of estimation errors on the HOCBF.
Since the RISE-based disturbance observer ensures the convergence of the estimation error, satisfying $\|\tilde{\bmd}(t)\|\leq\delta_1 e^{-\lambda_V t/2}$, the second term gradually vanishes over time.
As a result, the difference between $\Lambda(\x)$ and $L_{I}L_{\bar f}^{r-1}h(\x)\bmd$ converges to zero as $t\to\infty$.
Therefore, under the result of Lamma 2, $\Lambda(\x)$ serves as a valid solution to Problem\ref{problem1}.
Recall the HOCBF $h(\x)$ defined in Definition 2 and the RISE-based disturbance observer defined in (\ref{equ:obs}) and (\ref{equ:obs1}).
Define the candidate CBF as
\begin{align}
  \label{equ:barh}
  \bar h(\x) = \beta_1 \phi^{r-1}(\x)-V,
\end{align}
where $\beta_1$ is a parameter to be determined, $\phi^{r-1}(\x)$ and $V$ is defined in (\ref{equ:highorder_phi}) and (\ref{equ:V}), respectively.
\begin{theorem}
\label{theorem1}
Consider system (\ref{equ:sys2}) and the disturbance estimation law given in (\ref{equ:obs}) and (\ref{equ:obs1}).
Suppose that Assumption \ref{assumption:bounded_continuity} holds and $\phi^{i-1}(\x(0))>0, i\in[r]$.
If the parameters are selected such that $\beta>\delta_2+\delta_3/\max\{1,\alpha-\gamma\}$, $\beta_1\geq\delta_1^2/2\phi^{r-1}(\x(0))$, $\lambda_V\geq\lambda_r$,
then any Lipschitz continuous controller
\begin{align}
  \label{equ:khocbfdis}
  \bm u\in{\bar{\bm U}}_{{\rm H}}(\x)=\{&\bm u:L_{\bar f}^{r}h(\x)+L_{\bar g}L_{\bar f}^{r-1}h(\x)\bm u\nonumber\\
  &+\Lambda(\x)+\bm K^\top\bm\eta(\x)\geq0\}.
\end{align}
with $\bm K$ and $\bm \eta(\x)$ defined below (\ref{equ:highorder_limit_U1}), will guarantee $h(\x(t))\geq0$ for all $t\geq0$.
\end{theorem}
\begin{proof}
Considering the definition of ${\phi}^{r-1}(\x)$ and taking the time derivative of it yields
\begin{align}
\label{equ:dotphi}
\dot{\phi}^{r-1}(\x)=&L_{\bar f}^{r}h(\x)\!+\!L_{\bar g}\!L_{\bar f}^{r\!-\!1}\!h(\x)\bm u\!+\!L_{I}\!L_{\bar f}^{r\!-\!1}\!h^{r-1}(\x)\bmd\!\nonumber\\&+\!\bm K^\top\!\bm \eta(\x)-\lambda_r{\phi}^{r-1}(\x).
\end{align}
Then, substituting (\ref{equ:Vdot}) and (\ref{equ:dotphi}) into time derivative of (\ref{equ:barh}), one can obtain
\begin{align}
  \label{equ:dotbarh}
  \dot{\bar h}(\x) =& \beta_1 \dot{\phi}^{r-1}(\x)-\dot V\nonumber\\
\geq&\beta_1(L_{\bar f}^{r}h(\x)+L_{\bar g}L_{\bar f}^{r-1}\!h(\x)\bm u+L_{I}L_{\bar f}^{r-1}h(\x)\bmd\nonumber\\
&+\bm K^\top\!\bm \eta(\x)-\lambda_r{\phi}^{r-1}(\x))+\lambda_VV
\end{align}
Based on the definition of $\tilde{\bmd}$, the inequality (\ref{equ:dotbarh}) can be modified as
\begin{align}
  \label{equ:dotbarh1}
\dot{\bar h}(\x)\geq&\beta_1(L_{\bar f}^{r}h(\x)+L_{\bar g}L_{\bar f}^{r-1}\!h(\x)\bm u+L_{I}L_{\bar f}^{r-1}h(\x)\hat\bmd\nonumber\\
&-\|L_{I}L_{\bar f}^{r-1}h(\x)\|\|\tilde\bmd\|+\bm K^\top\!\bm \eta(\x))\nonumber\\
&-\lambda_r(\beta_1\phi^{r-1}(\x)-V)+(\lambda_V-\lambda_r)V,
\end{align}
Substituting (\ref{equ:dtilde1}) and (\ref{equ:lambdai}) into (\ref{equ:dotbarh1}) yields
\begin{align}
  \label{equ:dotbarh2}
  \dot{\bar h}(\x)\geq&\beta_1(L_{\bar f}^{r}h(\x)+L_{\bar g}L_{\bar f}^{r-1}\!h(\x)\bm u+\Lambda(\x)+\bm K^\top\!\bm \eta(\x))\nonumber\\
&-\lambda_r(\beta_1\phi^{r-1}(\x)-V)+(\lambda_V-\lambda_r)V\nonumber\\
\geq&-\lambda_r(\beta_1\phi^{r-1}(\x)-V)=-\lambda_r\bar h(\x),
\end{align}
Since $\beta_1\geq\delta_1^2/2\phi^{r-1}(\x(0))\geq\|\tilde{\bmd}\|^2/2\phi^{r-1}(\x(0))$ implies $\bar h(\x(0))\geq0$, (\ref{equ:dotbarh2}) indicates $\bar h(\x(t))\geq0$ for all $t\geq0$, such that $\phi^{r-1}(\x(t))\geq0$.
Because $\phi^{i-1}(\x(0))\geq0,i\in[r]$, according to Lamma 1, one has $h(\x(0))\geq0$.
\end{proof}

Building on the results of Theorem \ref{theorem1}, consider the control-affine system (\ref{equ:sys2}) with unknown disturbances, subject to $N$ HOCBF constraints $h_i(\x)$.
Refer to (\ref{equ:lambdai}), the time-varying lower bound for each $h_i(\x)$ can be defined as
\begin{align}
\Lambda_i(\x)\!\!=\!\!L_{I}\!L_{\bar f}^{r_i\!-\!1}h_i(\x)\hat\bmd\!-\!\|L_{I}\!L_{\bar f}^{r_i\!-\!1}h_i(\x)\|\delta_1 e^{\!-\!\lambda_V t/2},i\in[N].\nonumber
\end{align}
Then, the feasible space is denoted in the following form
\begin{align}
\bar{\bm U}_{\rm F}(\x) = \{\bm u\in \reals^m|\bm A(\x)\bm u\leq\bm b(\x)+\Delta\bm b(\x)\},\nonumber
\end{align}
where $\Delta \bm b(\x)=[\bm 0_{1\times m},\Lambda_1(\x),\dots,\Lambda_N(\x)]^\top$.
If $\bar{\bm U}_{\rm F}$ is nonempty, the robust safe controller can be obtained by solving the following DOB-CBF-QP:
\begin{subequations}
 \label{equ:cbfqp1}
\begin{align}
\label{equ:qp1}
  &\bm u^*(\x) = \arg\min\|\bm u-\bm u_{\rm ref}\|^2, \\
  \label{equ:constraints1}
  &{\rm{s.t.}}\quad\bm A(\x)\bm u\leq\bm b(\x)+\Delta\bm b(\x).
\end{align}
\end{subequations}
\subsection{DOB-VCBF-based Robust Safety Critical Control Framework}\label{sec:fea}
The constraints of the QP problem (\ref{equ:cbfqp1}) remain a convex polyhedron, which allows (\ref{equ:eillipsoidqp1}) to be used for constructing a VCBF in the form of (\ref{equ:Vv}).
However, under the influence of disturbances, the controller that satisfies (\ref{equ:limitD}) may no longer ensure the forward invariance of $\bm S_{V_v}$.
Motivated by (\ref{equ:lambdai}), the following lower bound is designed to cover the influence of the unknown disturbances
\begin{align}
\label{equ:lambdaVv}
\Lambda_{V_v}(\x)=L_{I}V_v(\x)\hat\bmd-\|L_{I}V_v(\x)\|\delta_1 e^{-\lambda_V t/2}.
\end{align}
Similarly, define the candidate CBF as
\begin{align}
\label{equ:barhV}
  \bar h_V(\x) = \beta_2 V_v(\x)-V,
\end{align}
where $\beta_2$ is a parameter to be determined. $V_v(\x)$ and $V$ is defined in (\ref{equ:Vv}) and (\ref{equ:V}), respectively.
\begin{theorem}
\label{theorem2}
Consider system (\ref{equ:sys2}) and the VCBF $V_v(\x)$ defined in (\ref{equ:eillipsoidqp1}).
Suppose that $V_v(\x(0))>0$.
If the parameters are selected such that $\beta>\delta_2+\delta_3/\max\{1,\alpha-\gamma\}$, $\beta_2\geq\delta_1^2/2V_v(\x(0))$, $\lambda_V\geq\lambda_{V_v}$.
Then any Lipschitz continuous controller
\begin{align}
  \label{equ:khocbfdis1}
  \bm u\in{\bar{\bm U}}_{{\rm C}}(\x)=\{&\bm u:L_{\bar f}V_v(\x)+L_{\bar g}V_v(\x)\bm u\nonumber\\
  &+\Lambda_{V_v}(\x)+\lambda_{V_v}V_v(\x)\geq0\}.
\end{align}
will guarantee $V_v(\x(t))\geq0$ for all $t\geq0$.
\end{theorem}
\begin{proof}
Taking the time derivative of (\ref{equ:barhV}) yields
\begin{align}
\label{equ:dotbarhV}
  \dot{\bar h}_V(\x) \geq& \beta_2 (L_{\bar f}V_v(\x)\!+\!L_{\bar g}V_v(\x)\bm u\!+\!L_{I}V_v(\x)\bmd)+\lambda_VV\nonumber\\
  \geq& \beta_2 (L_{\bar f}V_v(\x)\!+\!L_{\bar g}V_v(\x)\bm u\!+\!\Lambda_{V_v})\nonumber\\
  &+\lambda_{V_v}V+(\lambda_V-\lambda_{V_v})V,
\end{align}
Given the controller satisfying (\ref{equ:khocbfdis1}) and the parameter $\lambda_V\geq\lambda_{V_v}$, (\ref{equ:dotbarhV}) can be modified as
\begin{align}
\label{equ:dotbarhV1}
  \dot{\bar h}_V(\x) \geq&-\lambda_{V_v}(\beta_2 V_v(\x)-V)=-\lambda_{V_v}\bar{h}_V(\x)
\end{align}
Since $\beta_2\geq\delta_1^2/2V_v(\x(0))\geq\|\tilde{\bmd}\|^2/2V_v(\x(0))$ implies $\bar h_V(\x(0))\geq0$, (\ref{equ:dotbarhV1}) indicates $\bar h_V(\x(t))\geq0$ for all $t\geq0$.
In this way, $V_v(\x(t))\geq V/\beta_2\geq0$, which means that the CBF and input constraints is persistent compatible.
\end{proof}

However, the controller designed solely based on (\ref{equ:khocbfdis1}) fails to guarantee that $\bm u\in\bm U_{\rm F}(\x)$, which implies that the safety of the system trajectory might be violated.
Consequently, to enforce the compatibility of CBF and input constraints and ensure that the system remains within safe operational limits, the controller can be obtained by solving the following DOB-VCBF-QP:
\begin{subequations}
 \label{equ:qp}
\begin{align}
\label{equ:qp1}
  [\bm u^*(&\x), \delta^*(\x)]^\top =\! \arg\min\|\bm u\!-\!\bm u_{\rm ref}\|^2\!+\!M\delta^2, \\
  \label{equ:constraints1}
  {\rm{s.t.}}~&\bm A(\x)\bm u\leq \bm b(\x)+\Delta\bm b(\x),\\
  \label{equ:constraints2}
  &L\!_{\bar f}V\!_v(\x)\!+\!L_{\bar g}V_v(\x)\bm u\!+\!\Lambda_{V_v}(\x)\!\geq\!-\lambda_{V_v}V_v(\x)\!+\!\delta,
\end{align}
\end{subequations}
where $\delta$ is a slack variable that ensures the solvability of the QP as penalized by $M\gg 1$.
The relaxation is necessary because making that the feasible space non-empty is not equivalent to ensuring the safety of the system.
\begin{theorem}
\label{theorem3}
Suppose that $\bm u^*(\x)$ defined in (\ref{equ:qp}) is locally Lipschitz continuous.
Then, for any system trajectory $\x(t)$ starting from $\x(0)\in\bm S_H\bigcap\bm S_{V_v}$, if $\delta^*(\x(t))\equiv0$ then $\x(t)\in \bm S_H\bigcap\bm S_{V_v}$.
This guarantees that the system remains safe while ensuring mutual compatibility among the input constraints and multiple CBF constraints.
\end{theorem}
\begin{proof}
The solution of the closed-loop system is well-defined by the assumption that $\bm u^*(\x)$ is locally Lipschitz continuous.
When $\delta^*(\x(t))\equiv0$, the constraint (\ref{equ:constraints2}) can be rewritten as
\begin{align}
L\!_{\bar f}V\!_v(\x)\!+\!L_{\bar g}\!V_v(\x)\bm u\!+\!\Lambda_{V_v}(\x)\geq\!-\lambda_{V_v}V_v(\x), \nonumber
\end{align}
referring to Theorem \ref{theorem2}, it holds that $\x(t)\in\bm S_{V_v}$ for all $t\geq0$ if $\x(0)\in\bm S_{V_v}$.
The condition implies that the CBF and the input constraints are compatible and the feasible input space $\bm U_{\rm F}(\x)$ is nonempty.
Moreover, the constraint in (\ref{equ:constraints1}) is equivalent to the constraint conditions specified in the formula (\ref{equ:cbfqp1}).
This equivalence guarantees that $\bm u^*(\x)$ can ensure $\x(t) \in \bm S_H$ for all $t \geq 0$, if $\x(0) \in \bm S_H$. Therefore, the system's robust safety can be maintained even in the presence of disturbances.
By combining these results, $\bm u^*(\x)$ defined in (\ref{equ:qp}) can render $\x(t) \in \bm S_H\bigcap\bm S_{V_v}$ for all $t \geq 0$, if $\x(0) \in \bm S_H\bigcap\bm S_{V_v}$.
\end{proof}

\begin{figure}[t]
    \centering
    \begin{subfigure}[t]{0.8\linewidth} 
        \centering
        \includegraphics[width=\linewidth]{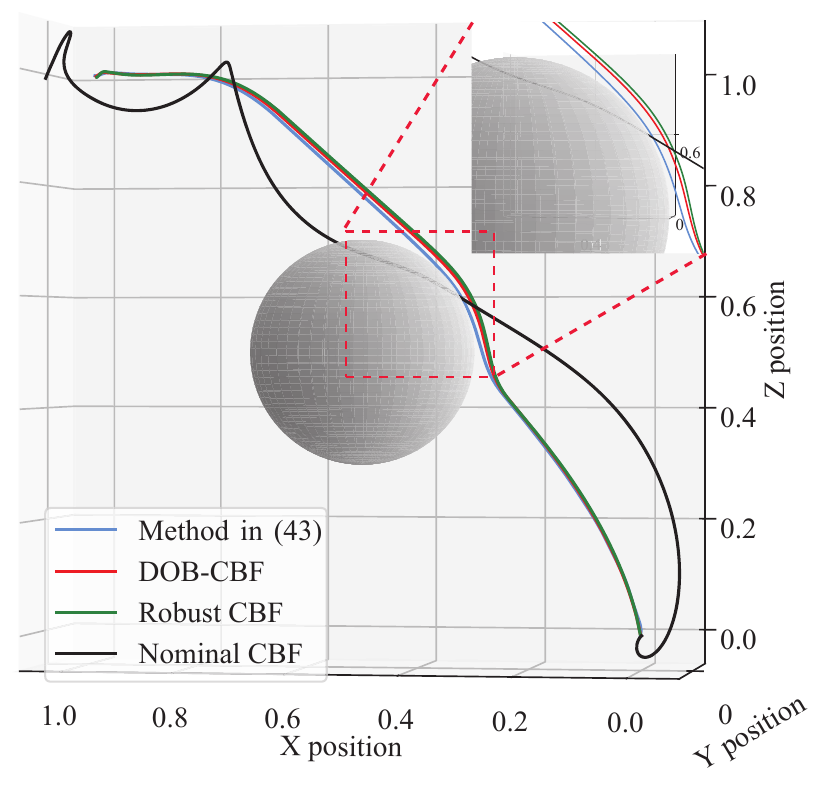} 
        \caption{Trajectories of the blimp.}
        \label{fig:sub_a}
    \end{subfigure}

    \begin{subfigure}[t]{0.48\linewidth} 
        \centering
        \includegraphics[width=\linewidth]{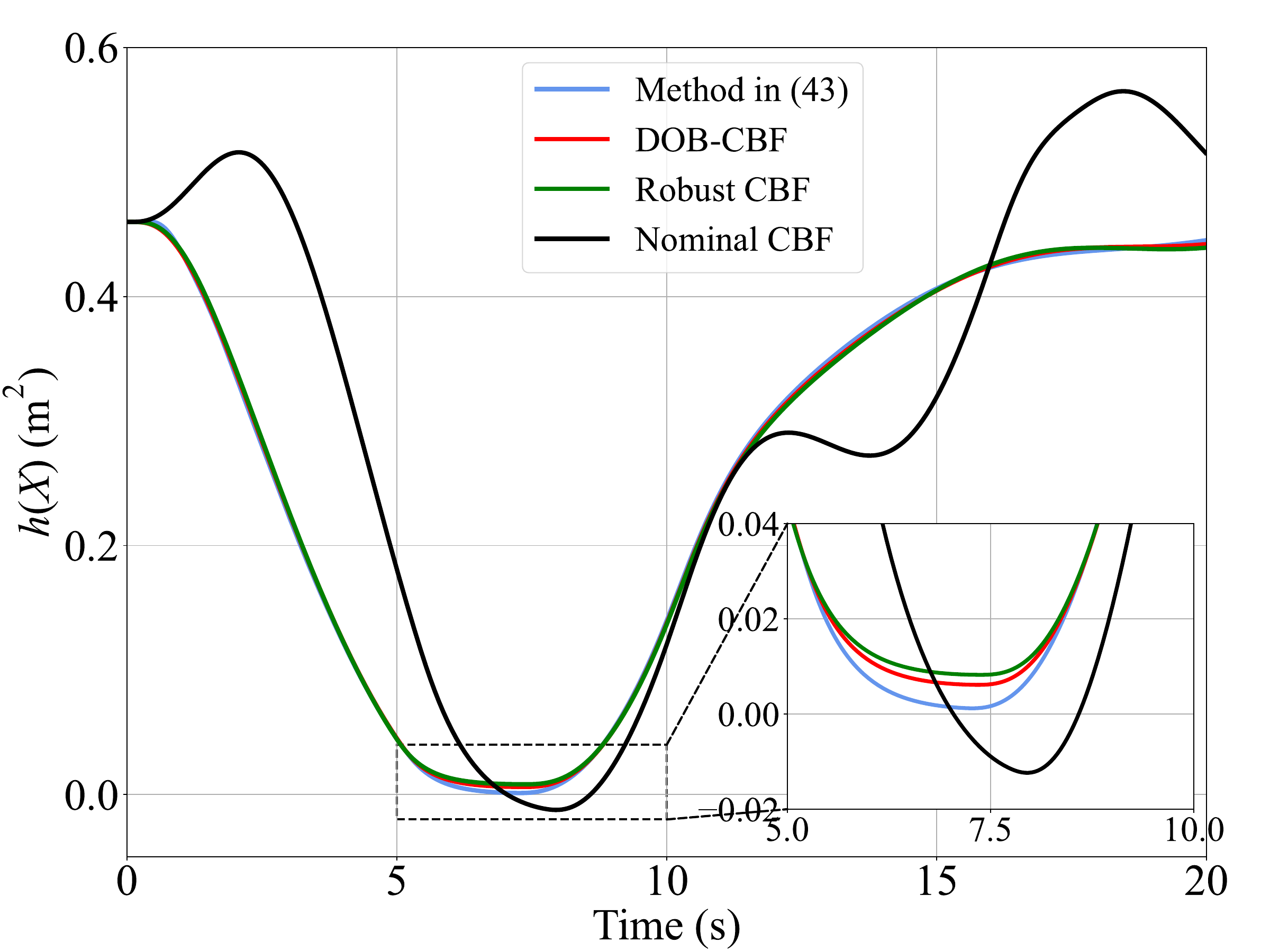} 
        \caption{$h(\bm X)$ of different comparison methods.}
        \label{fig:sub_b}
    \end{subfigure}
    \hfill
    \begin{subfigure}[t]{0.48\linewidth} 
        \centering
        \includegraphics[width=\linewidth]{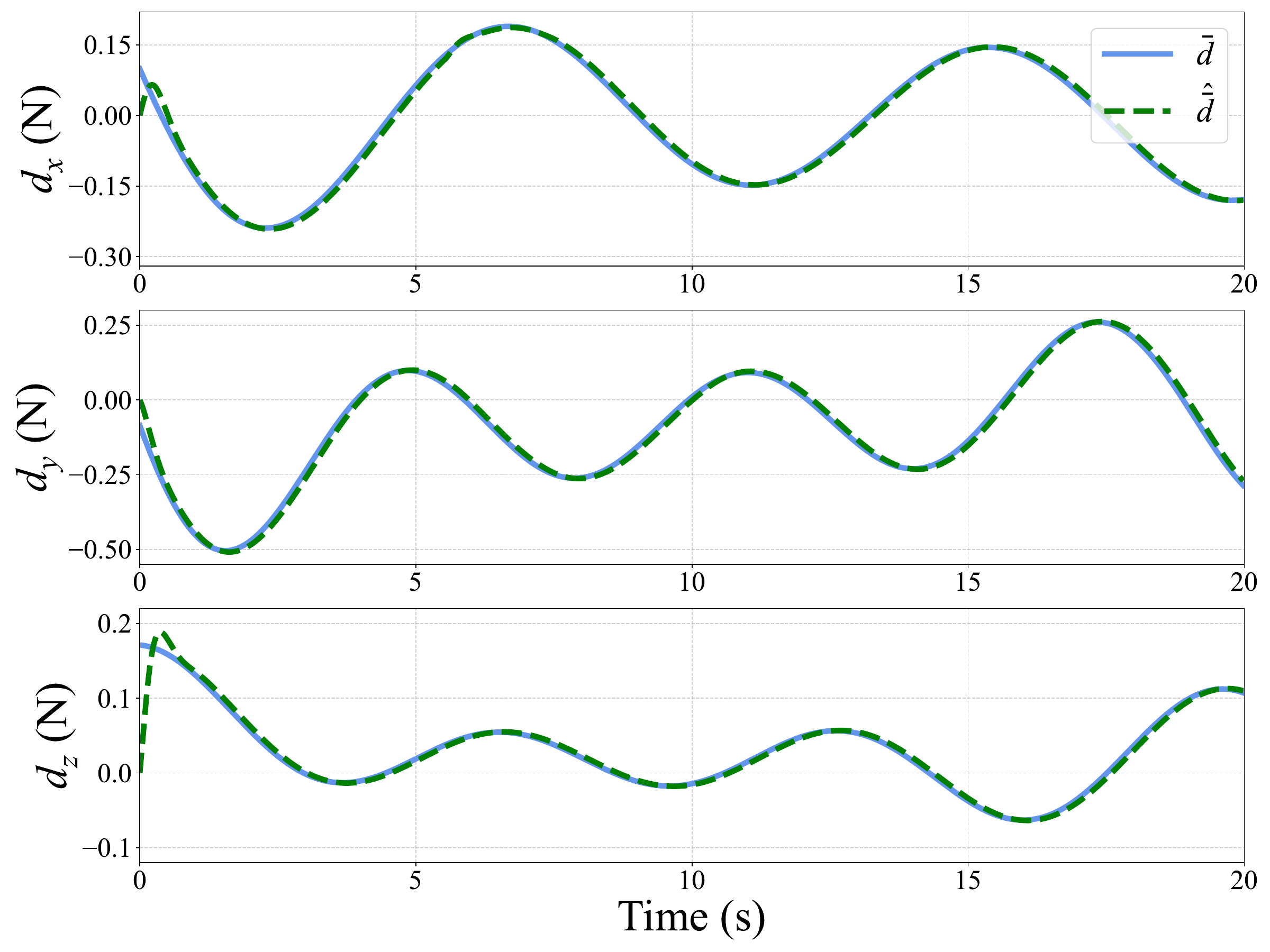} 
        \caption{Estimated disturbance $\hat{\bm d}$ of the RISE-based observer.}
        \label{fig:sub_c}
    \end{subfigure}

    \caption{Simulation results of blimp avoiding a single obstacle.}
    \label{fig:main}
    \vspace{-0.5cm}
\end{figure}

\begin{figure}[t]
    \centering
    \begin{subfigure}[t]{1.0\linewidth} 
        \centering
        \includegraphics[width=0.95\linewidth]{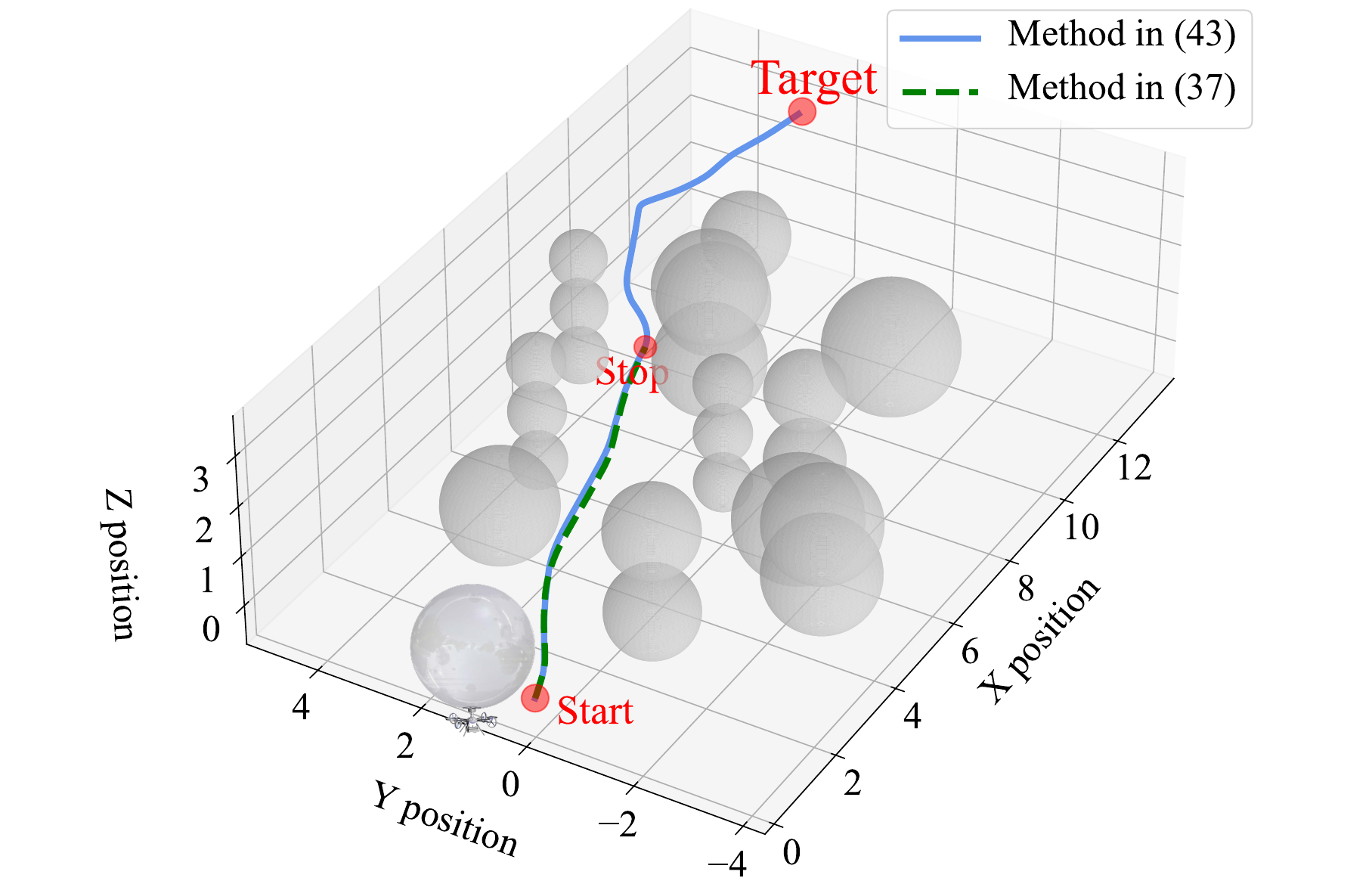} 
        \caption{Trajectories of the blimp.}
        \label{fig:sub_a}
    \end{subfigure}
    \begin{subfigure}[t]{0.43\linewidth} 
        \centering
        \includegraphics[width=\linewidth]{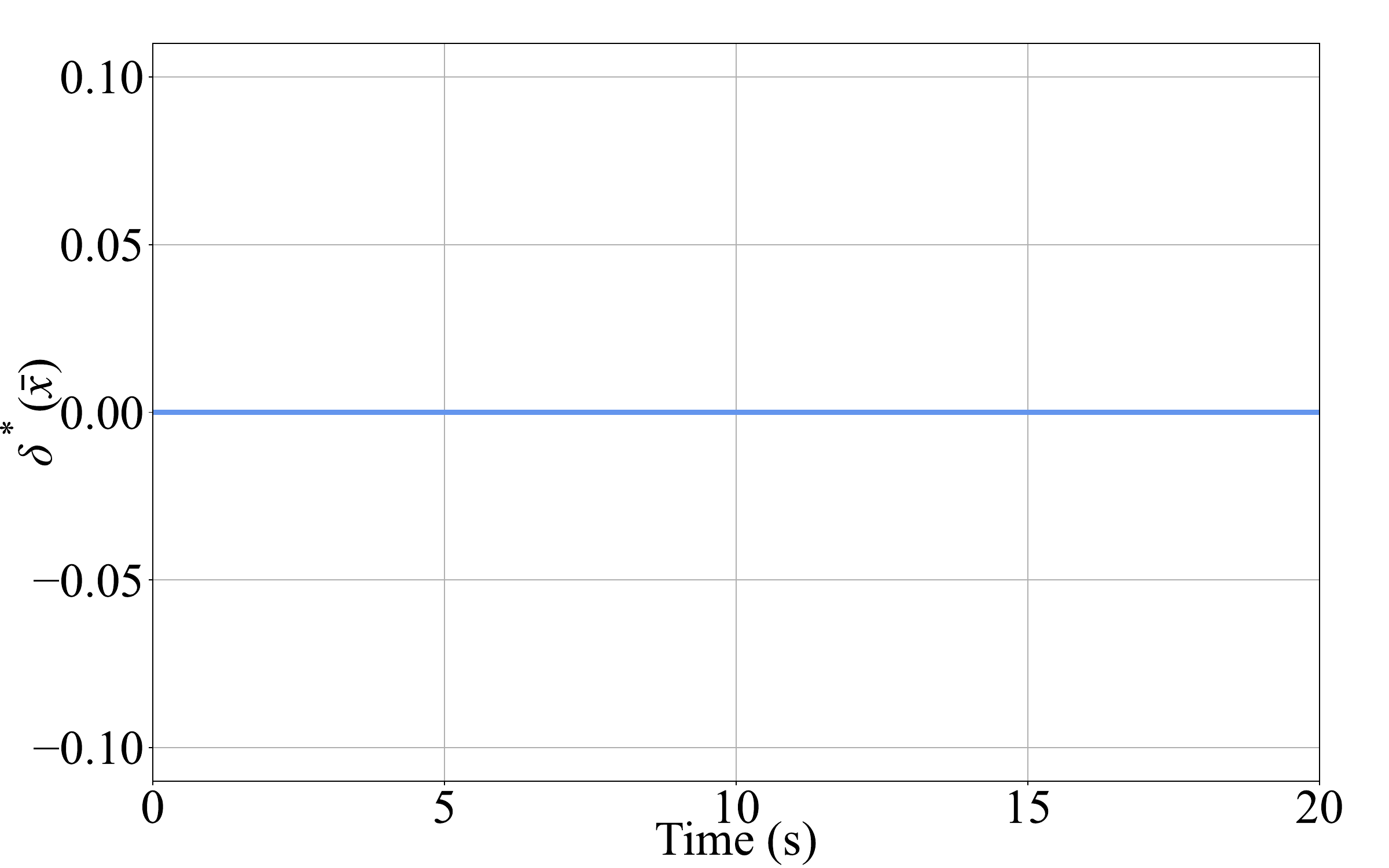} 
        \caption{$\delta^*(\x)$ of method in (\ref{equ:qp}).}
        \label{fig:sub_b}
    \end{subfigure}
    \hfill
    \begin{subfigure}[t]{0.43\linewidth} 
        \centering
        \includegraphics[width=\linewidth]{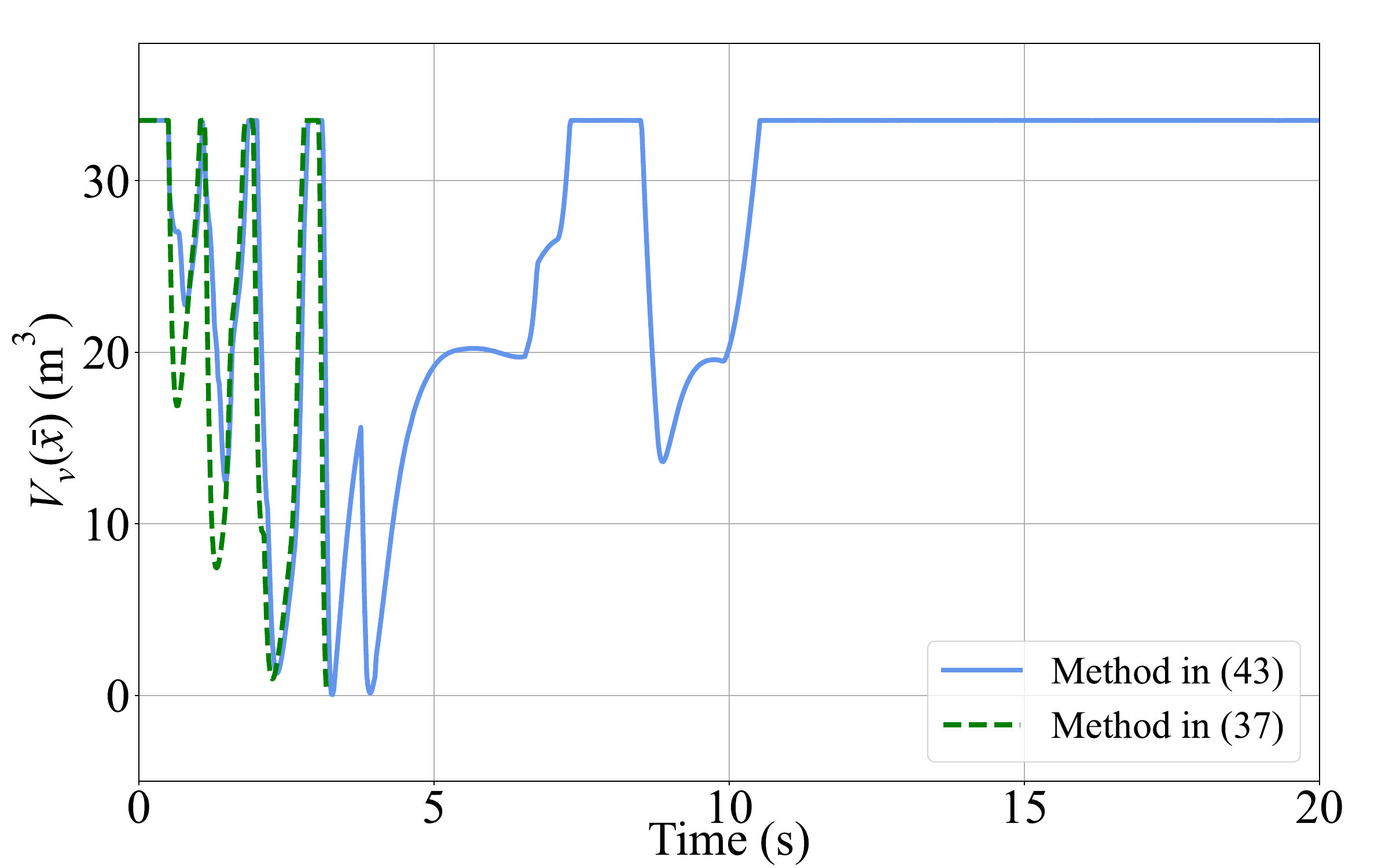} 
        \caption{$V_v(\x)$ of two methods.}
        \label{fig:sub_c}
    \end{subfigure}

    \begin{subfigure}[t]{0.45\linewidth} 
        \centering
        \includegraphics[width=\linewidth]{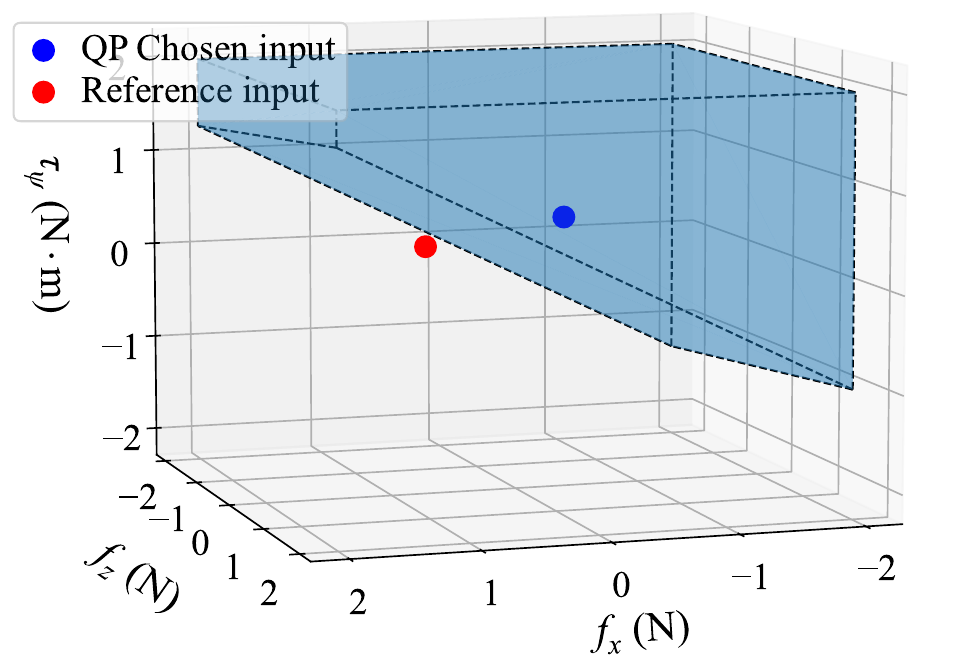} 
        \caption{Feasible space of method in (\ref{equ:qp}) at $t=2~\text{s}$.}
        \label{fig:sub_d}
    \end{subfigure}
    \hfill
    \begin{subfigure}[t]{0.45\linewidth} 
        \centering
        \includegraphics[width=\linewidth]{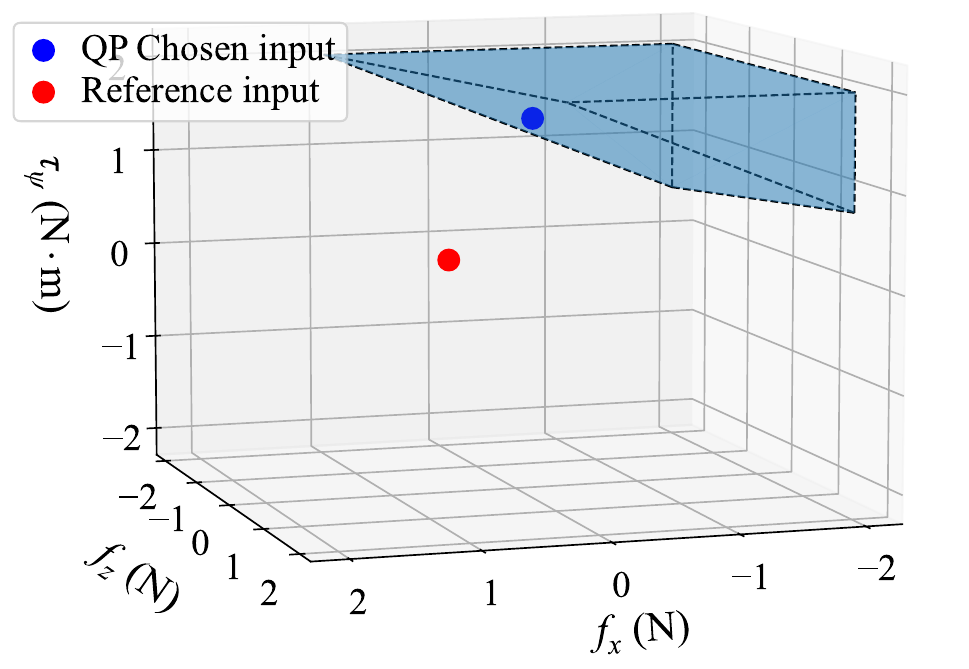} 
        \caption{Feasible space of method in (\ref{equ:cbfqp1}) at $t=2~\text{s}$.}
        \label{fig:sub_e}
    \end{subfigure}

    \begin{subfigure}[t]{0.45\linewidth} 
        \centering
        \includegraphics[width=\linewidth]{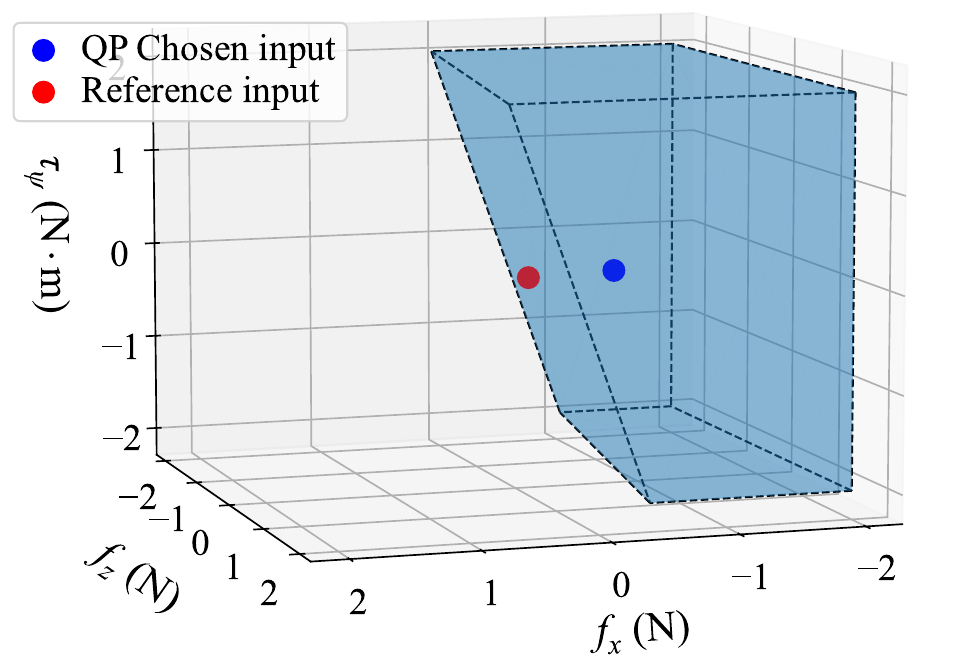} 
        \caption{Feasible space of method in (\ref{equ:qp}) at $t=3.2~\text{s}$.}
        \label{fig:sub_f}
    \end{subfigure}
    \hfill
    \begin{subfigure}[t]{0.45\linewidth} 
        \centering
        \includegraphics[width=\linewidth]{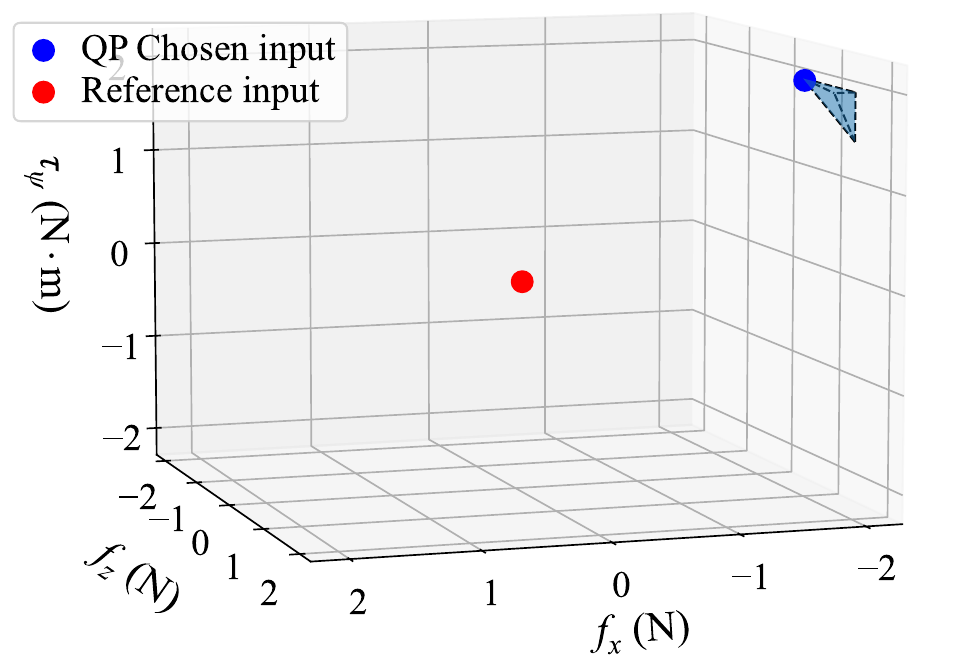} 
        \caption{Feasible space of method in (\ref{equ:cbfqp1}) at $t=3.2~\text{s}$.}
        \label{fig:sub_g}
    \end{subfigure}

    \caption{Simulation results of blimp avoiding multiple obstacles.}
    \label{fig:main1}
    \vspace{-5mm}
\end{figure}

\section{Simulation and Experimental Results}\label{sec:e}
\subsection{Simulation results}
To validate the effectiveness of the proposed control method, numerical simulations are conducted based on the underactuated blimp model from reference \cite{10679568}, where the system state $\bm X = [x, y, z, \psi]^\top \in \reals^4$ consists of the position coordinates $(x, y, z)$ and the yaw angle $\psi$, with control inputs consisting of the propulsion forces $f_x$ and $f_z$ along the $x$- and $z$- axes, respectively, and the yaw moment $\tau_\psi$ about the vertical axis.
In addition, force disturbances are considered in the $X$, $Y$, and $Z$ directions.
The blimp is expected to reach the target position $(x_{\text{t}},y_{\text{t}},z_{\text{t}})$ while avoiding collisions with spherical obstacles of varying sizes distributed in the state space. The feedback control law for the propulsion forces and yaw moment are given as follows:
\begin{align*}
f_x &= -k_v\left( v - k_x e_d \cos(e_\psi) \right) \\
f_z &= -k_w\left( \dot{z} + k_z e_z \right),
~N = -k_\psi e_\psi,
\end{align*}
where $v=\sqrt{\dot{x}^2+\dot{y}^2}, e_d=\sqrt{(x - x_{\text{t}})^2 + (y - y_{\text{t}})^2},e_\psi = \psi - \text{atan}(\frac{y_{\text{t}} - y}{x_{\text{t}} - x}), e_z = z - z_{\text{t}}$. The reference input incorporates the estimation results of the observer to enhance robustness and control performance.
The barrier function is defined as the signed distance to these spherical obstacles, that is, $h_i(\bm X)=d_i(\bm X, \bm S_i)$, where $d_i(\bm X, \bm S_i)=\sqrt{(x-x_i)^2+(y-y_i)^2+(z-z_i)^2}-R_i$ is the distance from the current state $\bm X$ to the sphere $\bm S_i$, which has a radius $R_i$ and is centered at $(x_i, y_i, z_i)$.  Since the input relative degree of $h_i(\bm X)$ is two, the HOCBF framework is applied to enforce safety constraints.
For this purpose, the parameters of HOCBF and VCBF are chosen as $k_i^1=10.0, k_i^2=3.55$, $\lambda_{V_v}=3.0$. To estimate the unknown disturbance, the parameters of the RISE-based observer are selected as $\alpha=5.0,\beta=6.0,\gamma=4.0$.

\emph{Test 1:} In this test, the proposed method is compared with the other there approaches (DOB-CBF\cite{10156095}, Robust CBF\cite{emam2019robust}, Nominal CBF (\ref{equ:cbfqp})) to validate its effectiveness in reducing conservativeness.
The blimp is required to move from the initial position $(0.0 ~\text{m},0.0 ~\text{m},0.0 ~\text{m})$ to the target position $(1.0 ~\text{m},0.0 ~\text{m},1.0 ~\text{m})$, avoiding an obstacle located at $(0.5~\text{m},0.0~\text{m},0.5~\text{m})$ with safety distance of $0.2~\text{m}$.
The results are shown in Fig. \ref{fig:main}.
Fig. \ref{fig:main}(\subref{fig:sub_a}) illustrates the trajectories of the blimp under different methods.
The proposed method, along with the robust methods presented in \cite{10156095} and \cite{emam2019robust}, can effectively handle disturbances, enabling the blimp to safely navigate around the obstacle and reach the target position.
In contrast, the nominal CBF fails to ensure the safety of the blimp.
The evolution of $h(\bm X)$ over time is shown in Fig. \ref{fig:main}(\subref{fig:sub_b}).
The blue curve represents the proposed method.
It is obvious that the value of $h(\bm X)$ of the proposed method is lower than those of the other methods while staying consistently above zero, indicating that the proposed method makes the blimp safe and has a lower conservativeness.
The simulation results in Fig. \ref{fig:main}(\subref{fig:sub_a}) and \ref{fig:main}(\subref{fig:sub_b}) demonstrate that the proposed method has a better tradeoff between performance and safety.
Fig. \ref{fig:main}(\subref{fig:sub_c}) shows the results of the RISE-based observer, demonstrating its ability to accurately estimate the disturbances.

\emph{Test 2:} In this test, the performance of the proposed VCBF is evaluated to demonstrate its ability to enhance the compatibility of candidate CBFs.
The blimp is expected to navigate around multiple obstacles in the state space from its initial position $(0.0 ~\text{m},0.0 ~\text{m},0.0 ~\text{m})$ and fly to the target position $(14.0 ~\text{m},3.0 ~\text{m},2.5 ~\text{m})$.
Fig. \ref{fig:main1} shows the performance of methods in DOB-VCBF-QP (\ref{equ:qp}) and DOB-CBF-QP (\ref{equ:cbfqp1}).
The trajectories of the blimp are illustrated in Fig. \ref{fig:main1}(\subref{fig:sub_a}).
In the presence of multiple obstacles, the method in (\ref{equ:qp}) ensures the safe arrival of the blimp at the target position, while the method in (\ref{equ:cbfqp1}) causes the blimp to stop after a certain period.
Fig. \ref{fig:main1}(\subref{fig:sub_b}) shows that $\delta^*(\x)\equiv0$, which, according to Theorem \ref{theorem3}, implies $V_v(\x)\geq0$ for all $t\geq0$.
This is further supported by Fig. \ref{fig:main1}(\subref{fig:sub_c}), where the method in (\ref{equ:qp}) consistently maintains higher $V_v(\x)$ values, ensuring feasibility and safety throughout the process.
Alternatively, the feasible space of the method in (\ref{equ:cbfqp1}) may vanish in scenarios with complex obstacle configurations.
This phenomenon can be observed in Fig. \ref{fig:main1}(\subref{fig:sub_g}).
Fig. \ref{fig:main1}(\subref{fig:sub_d})-(\subref{fig:sub_g}) provide the visualization of the feasible spaces for the two methods at different time instances. 
It can be seen that when the reference input is outside the feasible space, VCBF can map it to the interior of the feasible space.
This is done to prevent the volume of feasible space from shrinking too fast.
The method in (\ref{equ:cbfqp1}) maps it to a point on the boundary.

Besides, two simulation cases are designed to avoid the randomness of simulation results,
The robustness and effectiveness of the controller are evaluated through Monte Carlo simulations.
The initial states $(x_0, y_0, z_0)$ and the RISE-based observer parameters $(\alpha, \beta, \gamma)$ are randomized.
The average simulation runtimes, which reflect the performance under these randomized conditions, are summarized in Table \ref{tab:table1}.
The running time is defined as the smaller of two values: 20 s (maximum simulation duration) or the time when DOB-VCBF-QP ($T_1$) or DOB-CBF-QP ($T_2$) becomes infeasible.
As shown in the results, the performance of the method proposed in (\ref{equ:qp}) consistently outperforms that of the method in (\ref{equ:cbfqp1}) in both cases.

\begin{table}[t]
  \centering
  \caption{\small{Average Simulation Runtimes and Parameter Ranges}}
  \vspace{-0.1cm}
  \label{tab:table1}
  \setlength{\tabcolsep}{4pt} 
  \begin{tabular}{@{\hskip 1pt}c@{\hskip 1pt}|@{\hskip 1pt}c@{\hskip 2pt}c@{\hskip 2pt}c@{\hskip 2pt}c@{\hskip 2pt}c@{\hskip 2pt}c@{\hskip 2pt}c@{\hskip 2pt}c@{\hskip 1pt}}
    \hline
     & $T_1$ & $T_2$ & $x_0$ & $y_0$ & $z_0$ & $\alpha$ & $\beta$ & $\gamma$ \\ \hline
    Case 1    & \footnotesize 18.22 & \footnotesize 10.67 & \footnotesize [-0.1,0.1] & \footnotesize [-0.1,0.1] & \footnotesize [0.0,0.2] & \footnotesize 5.0 & \footnotesize 6.0 & \footnotesize 4.0 \\ \hline
    Case 2    & \footnotesize 16.21 & \footnotesize 8.67  & \footnotesize 0.0        & \footnotesize 0.0        & \footnotesize 0.0      & \footnotesize [4.6,6.0] & \footnotesize [5.0,7.0] & \footnotesize [3.6,4.4] \\ \hline
  \end{tabular}
  \vspace{-5mm}
\end{table}
\subsection{Experimental results}
In this section, the performance of the proposed method is tested on an Ackermann steering robot.
The purpose of this experiment is twofold: first, to verify that the proposed control method can ensure system safety by preventing collisions while exhibiting lower conservatism; and second, to demonstrate that the proposed method effectively reduces the sensitivity of the system to controller parameters.
The robot is modeled with position $\bm p=[x, y]^\top$, velocity $v$, heading angle $\psi$, with the specific dynamics $\dot x=v\cos(\psi), \dot y=v\sin(\psi), v=(a+d)/m, \dot \psi=\omega$. $m$ represents the mass of the robot. $d$ is the unknown disturbance.
The control inputs are $a$ and $\omega$. The robot is expected to reach the target position $\bm p_r=[x_r,y_r]^\top$ without any collisions.
The feedback controller is considered as follows:
\begin{align*}
&\psi_{r} = \text{atan}(\frac{{y_{r}} - y}{{x_{r}} - x}), ~e_{\psi} = \psi - \psi_{r}, ~ \omega = -k_{\omega} e_{\psi},\\
&v_{r} = k_p\sqrt{(x - {x_{r}})^2 \!+\! (y - {y_{r}})^2}\cos(e_{\psi}), a = -k_v (v - v_{r}),
\end{align*}
where $k_\omega,~k_p,~k_v>0$ are designed control parameters.
The reference input is further refined by incorporating the observer's estimation results into the feedback input to improve performance and robustness.
The barrier function is designed as $h_i=\bm d_i^\top\bm d_i-\bm d_{i_{\rm min}}^\top\bm d_{i_{\rm min}}$ with $\bm d_i$ being the distance of the robot to the center of the $i{\rm th}$ obstacle and $\bm d_{i_{\rm min}}$ being the safety distance of the obstacle.
\begin{figure}[t]
  \centering
  \includegraphics[width=3.4in]{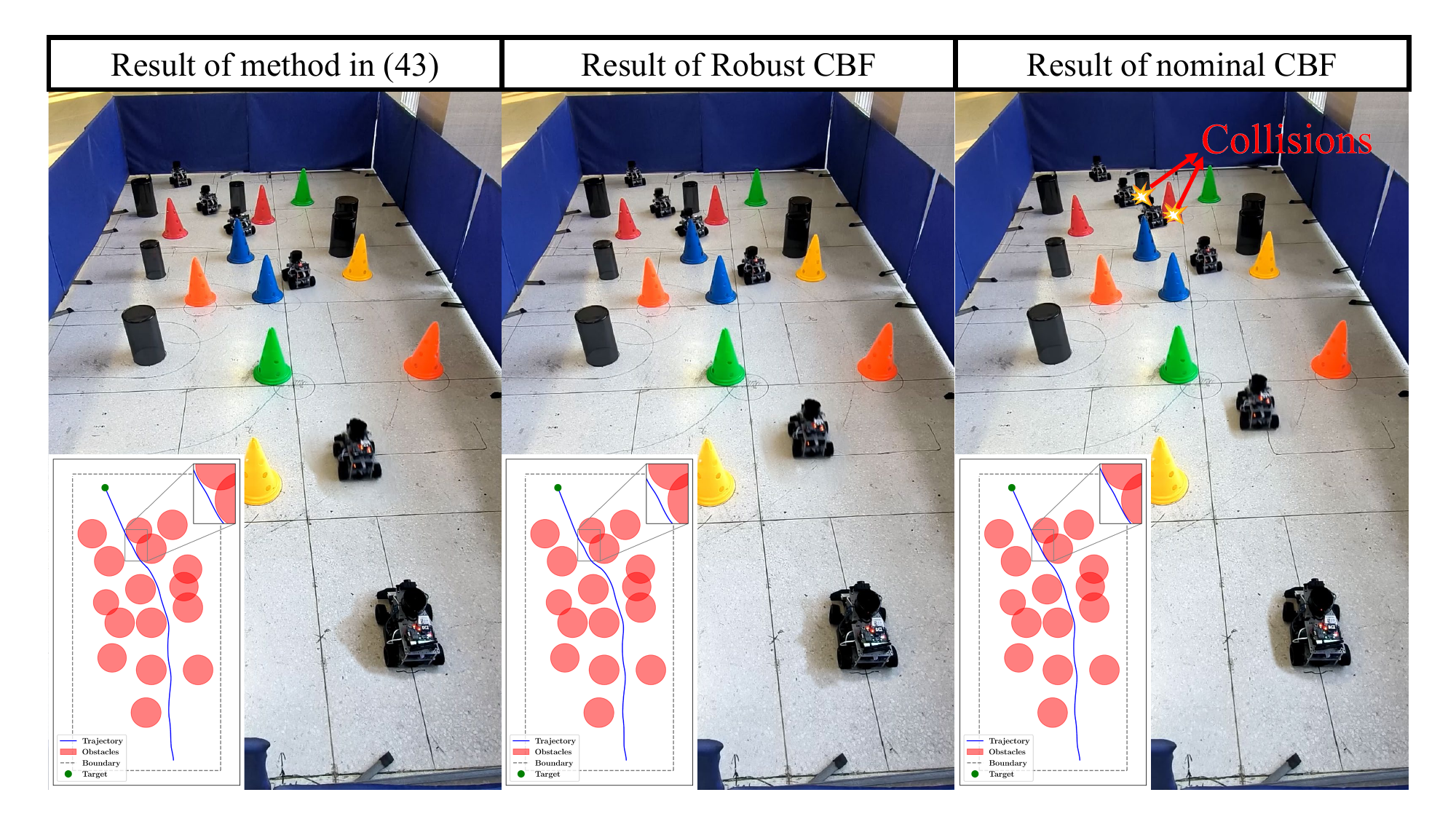}
  \captionsetup{font={small}}
  \caption{{ Trajectories to the target position with obstacle avoidance under different safety filters}}
  \label{fig:exp1}
  \vspace{-0.5cm}
\end{figure}

Fig. \ref{fig:exp1} illustrates how the controller in (\ref{equ:qp}) enables the Ackermann steering robot to avoid obstacles and reach the target position, along with comparative results from two methods proposed in \cite{emam2019robust} and \cite{tan2021high}.
The proposed method achieves an optimal balance between safety and efficiency by effectively avoiding obstacles with the required minimum distance.
It ensures collision avoidance while minimizing deviations from the reference input, avoiding unnecessary conservativeness.
In contrast, the robust CBF proposed in \cite{emam2019robust} enhances the system's robustness to external disturbances and avoids collisions with obstacles.
However, its conservativeness leads to trajectories with much larger distance from obstacles which may make it meet difficulties to reach target position especially in narrow environments.
The nominal CBF method \cite{tan2021high} fails to adjust the system state in a timely manner under external disturbances, ultimately leading to collisions.

After ensuring the robust safety of the controller, it is further verified that the controller (\ref{equ:qp}) can enhance the feasibility of the QP problem and reduce sensitivity to control parameters.
Five distinct values of $k_p$ in $\bm u_{\rm{ref}}$ is selected to implement sensitivity analysis.
The experiment results are shown in Fig. \ref{fig:exp2}.
It can be seen that for the controller with VCBF, regardless of the choice of $k_p$, the QP problem is always feasible, the robot can avoid collisions and efficiently reach the target position.
If VCBF is not used, when $k_p=0.2$, the robot's low reference velocity $v_r$ limits its ability to adjust its trajectory.
As a result, the controller prioritizes collision avoidance to ensure safety, neglecting the goal of reaching the target position.
This quickly renders the QP problem infeasible, causing the robot to stop near the obstacle to maintain safety.
This fact clearly demonstrates that the introduction of VCBF improves the feasibility of the QP problem and reduces sensitivity to control parameters.
The video showcasing the results of these simulations and experiments is available at \url{https://www.youtube.com/watch?v=0gfpCG5fKp0}.
\begin{figure}[t]
  \centering
  \includegraphics[width=3.4in]{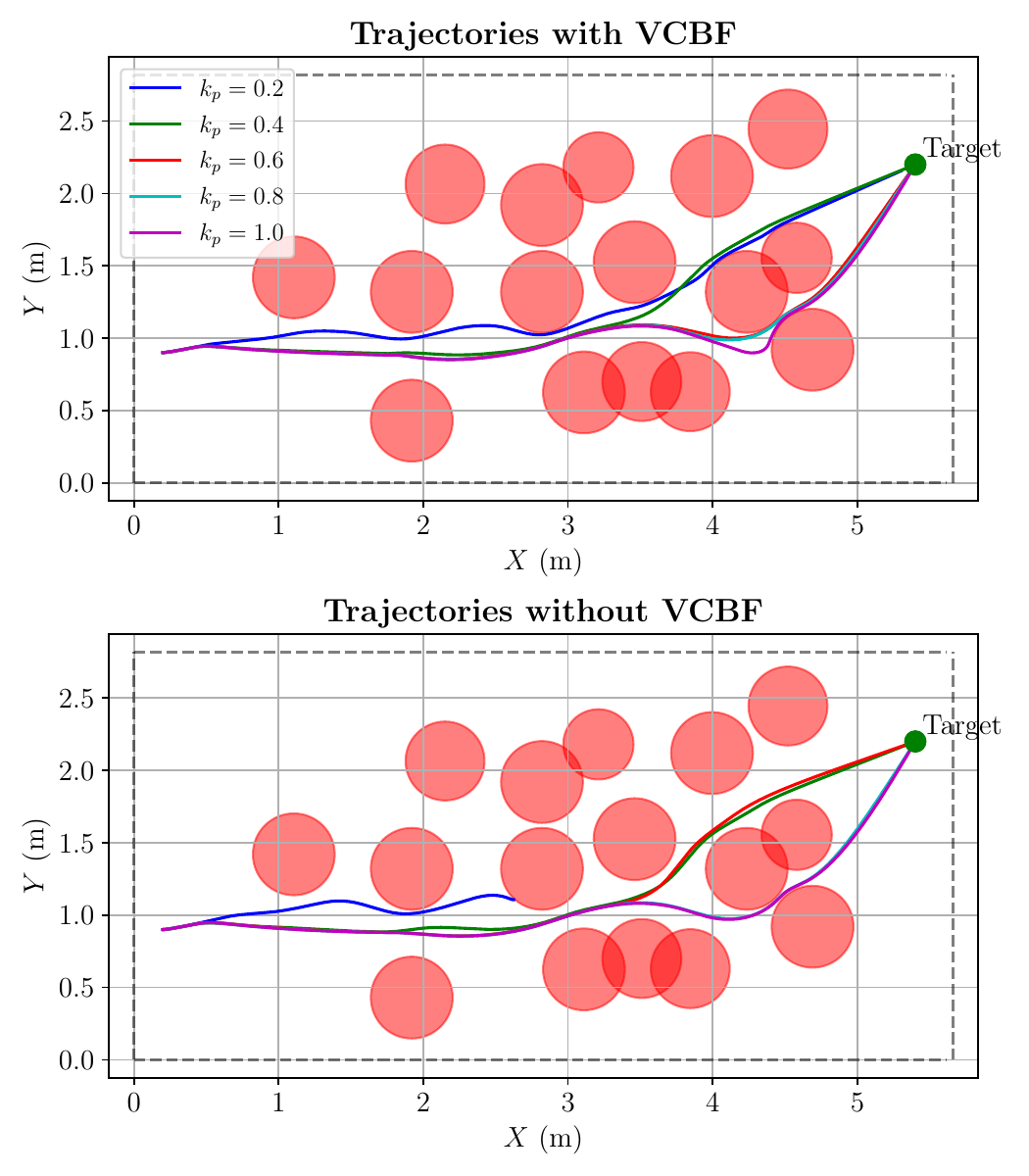}
  \vspace{-2mm}
  \captionsetup{font={small}}
  \caption{{ Trajectories for the five cases of $k_p = 0.2, 0.4, 0.6, 0.8, 1.0$.}}
  \label{fig:exp2}
  \vspace{-5mm}
\end{figure}
\section{Conclusions}\label{sec:f}
This paper proposes a robust control framework for uncertain dynamical systems with multiple CBF and input constraints, ensuring safety and maintaining system performance.
A RISE-based observer is used to estimate uncertainties, with its error bound integrated into the controller to reduce conservativeness.
To address conflicts among multiple CBFs, the VCBF is introduced to analyze and preserve the feasible space of the QP problem under disturbances.
Furthermore, a DOB-VCBF-QP based control law is developed to ensure system safety while maintaining compatibility between CBF constraints and input constraints.
Simulations and experiments validate the effectiveness of the proposed method in achieving safety and robust performance under uncertainty.
Future research will focus on exploring the conditions necessary to maintain both compatibility and safety simultaneously.
\end{spacing}
\begin{spacing}{1}
\bibliographystyle{Bibliography/IEEEtranTIE}
\bibliography{Bibliography/BIB_xx-TIE-xxxx}
\end{spacing}

\end{document}